\documentclass{amsart}

\usepackage{amssymb}
\usepackage{enumerate}
\usepackage{amsmath}
\usepackage{amsthm}
\usepackage{textcomp}

\usepackage{graphicx,epsfig,tikz}
\definecolor{wrwrwr}{rgb}{0.3803921568627451,0.3803921568627451,0.3803921568627451}
\definecolor{dtsfsf}{rgb}{0.8274509803921568,0.1843137254901961,0.1843137254901961}
\definecolor{sexdts}{rgb}{0.1803921568627451,0.49019607843137253,0.19607843137254902}
\definecolor{rvwvcq}{rgb}{0.08235294117647059,0.396078431372549,0.7529411764705882}
\definecolor{cqcqcq}{rgb}{0.7529411764705882,0.7529411764705882,0.7529411764705882}

\usepackage{amsmath,amsfonts,amssymb,enumerate,mdwlist}
\usepackage{latexsym,amsfonts,amssymb,amscd,eufrak}
\usepackage{fancybox,shadow,color,graphicx,psfrag,float}

\usepackage{stackengine}

\newcommand{\m}{\mathcal}

\newcommand{\F}{\mathbb{F}}
\renewcommand{\P}{\mathbb{P}}
\newcommand{\divisor}{\operatorname{Div}}

\newcommand{\con}{\operatorname{Con}}

\newcommand{\supp}{\operatorname{Supp}}

\newtheorem{thm}{Theorem}[section]
\newtheorem{pro}[thm]{Proposition}
\newtheorem{lem}[thm]{Lemma}
\newtheorem{cor}[thm]{Corollary}
\newtheorem{conjecture}[thm]{Conjecture}
\newtheorem{rem}[thm]{Remark}

\theoremstyle{definition}
\newtheorem{defin}{Definition}
\newtheorem{example}{Example}

\begin{document}

\title{Locally recoverable codes from towers of function fields} %

\author{M. Chara, F. Galluccio and E. Martínez-Moro}
     \email{M. Chara: mchara@santafe-conicet.gov.ar}
     \email{F. Galluccio: frangallu996@gmail.com}
         \email{E. Martínez-Moro: edgar.martinez@uva.es}




\thanks{This work is partially supported by CONICET, ANPCyT PICT, UNL CAI+D 2020. First autor was partially funded by Research in Pairs Fellowship - CIMPA 2021. Second author was partially funded by a doctoral grant of the program Erasmus+ KA 107. Third author was supported in part by Grant PGC2018-096446-B-C21 funded by MCIN/AEI/
	10.13039/501100011033 and by “ERDF A way of making Europe”.\\
M. Chara: Researcher of CONICET at FIQ, Universidad Nacional del Litoral, Argentina\\
F. Galluccio: Universidad Nacional del Litoral and CONICET, Argentina\\
E. Martínez-Moro: Institute of Mathematics, Universidad de Valladolid, Spain}

\begin{abstract}

In this work we construct sequences of locally recoverable AG codes arising from a tower of function fields and give bound for the parameters of the obtained codes. In a particular case of a tower over $\F_{q^2}$ for any odd $q$, defined by Garcia and Stichtenoth in \cite{GS}, we show that the bound is sharp for the first code in the sequence, and we include a detailed analysis for the following codes in the sequence based on the distribution of rational places that split completely in the considered function field extension.

\bigskip
\noindent{\it Key words: Function fields, Towers, Codes, LRC, Asymptotic behavior}

\bigskip
\noindent{\it 2020 Mathematical Subject Classification: 94B27, 14H05, 11G20, 11T71}
\end{abstract}
\maketitle

\section{Introduction}

 Let $q$ be a prime power.  A linear code $\m{C}$ of length $n$ over the finite field  $\F_q$ with $q$ elements is a linear subspace of $\F_q^n$. The dimension of the code $k$ is the dimension of $\m{C}$ as a $\F_q$-vector space and its minimum distance $d$ is the minimum of the number of non-zero entries of a vector in $\m{C}\setminus\{\mathbf 0\}$. A linear code $\m{C}$ of length $n$ and dimension $k$  is a locally recoverable error-correcting code, (LRC for short) with locality $r$ if it is a code that produces an $n$-coordinate codeword from $k$ information symbols and, for any coordinate of the codeword, there exist at most $r$ other coordinates that allow to recover the lost coordinate. LRCs were introduced in \cite{GHSY12} motivated by the significant use of coding techniques applied to data distribution and cloud storage systems, since local recovery techniques enable us to repair lost encoded data by making use of small amount of data within the received vector instead of all information. 

One of the problems of interest in the field is the construction of long non trivial codes, in which the ground field cardinality is not much larger than the code length. It is well known that one way of obtaining long codes is by the use of algebraic curves with many rational points, or equivalently, algebraic function fields with many rational places. In this work we are interested in working with this type of algebraic-geometry codes (AG codes for short), or codes coming from the evaluation of rational points on an algebraic curve over $\F_q$, (see Section \ref{Pre} for precise definitions).  In this direction, Barg, Tamo and Vladut \cite{BargTamoVladut} constructed algebraic-geometry locally recoverable codes, generalizing previous works of Barg and Tamo \cite{TamoBarg, TamoBargOptimal}. Haymaker, Malmskog and Matthews in \cite{HMM}, used fiber products of maximal curves to construct LRC with availability, giving in this way a generalization of \cite{BargTamoVladut} for more than one recovery set. 

Since we aim to build sequences of AG codes that are also LR codes, we will work with towers $\m{F}=(F_0, F_1, \ldots)$ of function fields over the same finite field $\F_q$, in the way defined by Garcia and Stichtenoth (see for example \cite{GS, stichtenoth}). 
By using these type of asymptotically optimal towers, Tamo, Barg and  Frolov in \cite{TBF16} gave a construction of asymptotically good locally recoverable codes over $\F_{q^2}$ with locality $r=q-1$ whose rate $R$ and relative distance $\delta$ satisfy the following equation
\begin{equation}\label{GVthem}
R\geq \frac{r}{r+1}\left(1-\delta-\frac{3}{q+1}\right).\end{equation}

In our work, we also employ the Garcia-Stichtenoth function field given in \cite{GS} to build asymptotically good LRCs. Our method can be viewed as an extension of the construction given in \cite{BargTamoVladut}. Furthermore, our construction allows to compute the exact dimension of the code in severals steps on the tower, using a different evaluation set but in this case, although we can improve the dimension, we have a cost to pay in the bound for the minimum distance. Nevertheless, we are able to construct a family of LRC over 
$\F_{q^2}$ with locality $r=q-1$ whose rate $R$ and relative distance $\delta$ satisfy
\begin{equation}\label{GVus}R\geq \frac{r}{r+1}\left(1-\delta-\frac{2}{q}\right),\end{equation} which improves the asymptotic Gilbert-Varshamov type bound for LRCs with a single recovery given by \eqref{GVthem}.
This bound is a particular case of the bound given by Li, Ma and Xing in \cite{LMX19}, in which they use the same tower and the automorphism group of the function fields involved to build asymptotically good LRCs. Nevertheless, note that the main difference with this work is that we are able to compute the exact dimension of the constructed codes.

This work is organized as follows. In Section~\ref{Pre} we give some preliminary definitions and facts on the behavior of towers of function fields. Section~\ref{Primera} is devoted  to build locally recoverable codes from a tower of function fields, see Theorem~\ref{Teounr}.
Determining the exact minimum distance  of linear codes can be a hard problem. In particular, few results are known for the exact minimum distance of LRCs. In \cite{CKMTW}, this problem is addressed by using the construction given in \cite{HMM}, and analyzing relative parameters. In Section~\ref{sec4}, and in particular in Propositions \ref{exactmd} and \ref{prop42}, we compute the exact minimum distance of the constructed codes in one step extension of the tower, and we show that the method used to compute this distance will not be enough to compute the minimum distance in further steps. Finally, in the last section of this paper, we compute some asymptotic parameters and show the improvements mentioned in Equation \eqref{GVus}.

 \section{Preliminaries}\label{Pre}

Throughout this work, we will use the terminology of function fields for defining AG-codes (instead of algebraic curves).  An algebraic function field $F$ over $\F_q$ is a finite algebraic extension  of the rational function field $\F_q(x)$, where $x$ is a transcendental element over $\F_q$. For a complete introduction on  algebraic function field and their relationship with codes see for example~\cite{stichtenoth}.

Let $D$ and $S=P_1+\cdots +P_n$ be two disjoint divisors of an algebraic function field $F$ over $\F_q$, 
where $P_1,\ldots,P_n$ are different rational places of $F$, and consider the Riemann-Roch space associated to $D$
$$\m{L}(D) = \{f\in F \setminus \{0\} : (f) \geq -D\} \cup \{0\},$$
 where $(f)$ denotes the principal divisor of $f\in F$. The algebraic-geometry code
defined by $F$, $D$ and $S$ is
$$\m{C} = C_{\m{L}}(S,D) = \big \{ \big( f(P_1),f(P_2),\ldots, f(P_n) \big) \in \F_q^n : f\in \m{L}(D) \big\},$$
where $f(P_i)$ stands for the residue class of $f$ modulo $P_i$, $i=1,\dots, n$.

Since our  aim is to build infinite sequences of AG codes that are also LR codes, we will work with infinite sequences of function fields. For fixing our notation and settings we will use the following definitions.

\begin{defin}
	 We will consider sequences $\m{F}=(F_0, F_1, \ldots)$, of function fields over the same finite field $\F_q$ such that
	\begin{enumerate}[1.]
		\item $F_0\subsetneq F_1 \subsetneq F_2 \subsetneq \cdots$,
		\item the field extension $F_{i+1}/F_i$ is finite and separable for all $i\geq 0$, and
		\item $\F_q$ is the full constant field of each $F_i$ for all $i\geq 0$, i.e., the only elements of $F_i$ which are algebraic over $\F_q$ are the elements of $\F_q$.
	\end{enumerate}
Following the works of Garcia and Stichtenoth (see for example \cite{GS}, \cite{stichtenoth}), we will say that the sequence $\m{F}$ is a \emph{tower of function fields over} $\F_q$, if the genus $g(F_i)$ grows to infinity as $i\rightarrow \infty$.
\end{defin}

\begin{defin}
A sequence (or a tower) $\m{F}=(F_0, F_1, \ldots)$ is \emph{recursively defined} if there exist a bivariate polynomial $f\in \F_q[S,T]$ and transcendental elements $x_i$ $i=1,2,\ldots$, such that for all $i\geq0$ the following statements hold:
\begin{enumerate}[1.]
\item $F_0=\F_q(x_0)$ is the rational function field.
\item $F_{i+1}=F_i(x_{i+1})$ with $f(x_i, x_{i+1})=0$.
\item $[F_{i+1}:F_i]=\deg_{T}f$.
\end{enumerate}
\end{defin}

 The following definitions are important in the study of the asymptotic behavior of sequences of function fields.  As usual, we will denote by $\P(F)$ the set of all places of the function field $F/\F_q$. 
  Given a finite extension $E/F$ and a place $P\in \P(F)$ there are finitely many places $Q\in \P(E)$ lying above $P$. We will write $Q|P$ when $Q$ lies over $P$. If $\m{F}=(F_0, F_1, \ldots)$ is a sequence of function fields over $\F_q$, we say that a place $P\in \mathbb{P}(F_i)$ {\em splits completely} in $\mathcal{F}$ if $P$ splits completely in each extension $F_j/F_i$ (and in this case we have $[F_j:F_i]$ different places over $P$). The {\em splitting locus} of $\mathcal{F}$ over $F_0$ is defined as
\[\mathrm{Split}(\mathcal{F}/F_0):=\{P\in \mathbb{P}(F_0)\,:\,\text{$\deg(P)=1$ and $P$ splits completely in $\mathcal{F}$}\}\,,\] where $\deg(P)$ is the degree of the place $P$.
A place $P\in \mathbb{P}(F_i)$ is {\em ramified} in $\mathcal{F}$ if $P$ is ramified in any extension $F_j/F_i$.
The {\em ramification locus} of $\m{F}$ over $F_0$ is the set
$$\mathrm{Ram}(\m{F}/F_0):=\{P\in \mathbb{P}(F_0)\,:\,\text{$P$ ramified in some extension }F_n/F_0\}.$$ A place $P\in \mathbb{P}(F_i)$ {\em is totally ramified} in $\mathcal{F}$ if $P$ is totally ramified in each extension $F_j/F_i$ (in this case we have only one place $Q$ in $\P(F_j)$ over $P$ and the ramification index $e(Q|P)$ is equal to $[F_j:F_i]$).  The {\em complete ramification locus} of $\mathcal{F}$ over $F_0$ is defined as
\[\mathrm{CRam}(\mathcal{F}/F_0) :=\{P\in \mathbb{P}(F_0)\,:\,\text{$\deg(P)=1$ and $P$ is totally ramified in $\mathcal{F}$}\}\,.\]
Since every place $Q\in \mathbb{P}(F_i)$ lying above a place in $\mathrm{Split}(\mathcal{F}/F_0)\cup \mathrm{CRam}(\mathcal{F}/F_0)$ is a rational place (i.e. of degree one), we have that
\begin{equation}\label{splitcram}
N(F_i)\geq [F_i:F_0]|\mathrm{Split}(\mathcal{F}/F_0)|+|\mathrm{CRam}(\mathcal{F}/F_0)|\,,
\end{equation} where $N(F_i)$ is the number of rational places of $F_i$.

\begin{lem}
If $\mathcal{F}$ is a sequence such that $\mathrm{Split}(\mathcal{F}/F_0)\neq \emptyset$ then $\mathcal{F}$ is a tower.
\end{lem}
\begin{proof}
Since  $\mathrm{Split}(\mathcal{F}/F_0)\neq \emptyset$  there is a rational place  $P$ in $F_0$ that splits completely in each extension $F_i/F_0$ then, by the Hasse-Weil bound, we have that $g(F_i)\rightarrow\infty$ as  $i\rightarrow\infty$  so that $\mathcal{F}$ is actually a tower. 
\end{proof}

The \emph{limit} of a tower of function fields $\m{F}$ is given  by $$\lambda(\m{F})=\frac{N(F)}{g(F)},$$ it always exist and it is a non negative amount, so it is said that the tower is \emph{asymptotically good} if $\lambda(\m{F})>0$ and \emph{asymptotically bad} on the other case. Therefore, a tower is asymptotically good if and only if it has non-emply splitting locus and finite ramification locus (see \cite[Proposition 7.2.6]{stichtenoth}).

In this work we will ``lift'' divisors from one function field $F$ to an extension $E$ by using the conorm map on divisors 
$$\con_{E/F} : \divisor(F) \rightarrow \divisor(E),$$ 
that we now recall.  
If $P$ is a place in $F$, the conorm divisor of $P$ is the divisor 
$$\con_{E/F}(P) = \sum_{_{Q|P}} e(Q|P) \,Q$$ 
in $E$, where $e(Q|P)$ is the ramification index of the place $Q$ in $E$ over $P$. 
Now, the \textit{conorm divisor} of $D=\sum_P n_P P \in \divisor(F)$ in $E$ is given by
\begin{equation*} \label{conorm}
D':= \con_{E/F} (D) = \sum_P n_P \con_{E/F}(P). 
\end{equation*}

\section{Sequences of LRC}\label{Primera}

First we will present the construction of a LR code on an extension of function fields. This construction can be found Theorem 3.1 in \cite{BargTamoVladut}. We included here for the sake of completeness and adapting it to the function field language.  

\begin{thm}\label{teobargtamo} Let $F$ a function field over $\F_q$  and let $E=F(x)$ a function field extension of degree $m$. Let $S$ be a set of places of $F$ that split completely in $E/F$ and such that  
$$\{Q \in \P(E): Q\cap F \in S\}
\cap\{Q\in \P(E): \nu_Q(x)<0\}=\emptyset$$ and let $\m{B}=\{Q\in \P(E): Q\cap F\in S \}$. Then, if $|S|=s>0$ we have that $|\m{B}|= sm$. Choose a divisor $D$ of $F$ of degree $l$ such that $\supp(D)\cap S=\emptyset$, and let $\{f_1, \ldots, f_\ell\}$ a basis for $\m{L}(D)$, the Riemann-Roch space associated to $D$. Let $r=m-1$ and consider the space $V$ generated by 
$$\{f_wx^{e}:w=1, \ldots, \ell; e=0,\ldots, r-1\}.$$ Since $\supp(D)$ is disjoint from S, the evaluation map 
$$
\begin{array}{cccl}
 ev: & V  &  \longrightarrow& \F_q^{(r+1)s} \\
 & f& \rightarrow  &  (f(P_{11}), \ldots, f(P_{ms}) )
\end{array}
$$ is well defined. The image of this mapping is a locally recoverable code $\m{C}$ with locality $r$, which we denote by
$C(S,D)$. The code coordinates are naturally partitioned into $s$ subsets $A_j = \{P_{ij}\}_{i = 1,...,r+1}$ of size $r+1$ each. Denoting by $g(F)$ the genus of $F$ and $h=[E:\F_q(x)]$, we have that the parameters of the code satisfy

\begin{itemize}
\item $n= (r+1)s$,
 \item $k= r\ell\geq r(l+1-g(F))$,
 \item $d\geq n-l(r+1)-(r-1)h$,
\end{itemize}
 provided that the right-hand side of the inequality for $d$ is a positive integer. Local recovery of an erased symbol $c_{ij} = f (P_{ij})$ can be performed by polynomial interpolation through the positions of the points in the recovery set $A_j\setminus\{P_{ij}\}$.\end{thm}

\begin{rem}
Notice that the previous construction can be performed for any step in a tower of function fields, either consecutive or not. 
\end{rem}

In particular, we can consider a tower of function fields, and built a LR code, but starting with the rational function field, as the base field in the tower. 

\begin{thm} Let $\m{F}=\{F_j\}_{j=0}^{\infty}$ be a sequence of function fields and $\{x_i\}_{j=0}^{\infty}$ a sequence of transcendental elements over $\F_q$ such that $F_0=\F_q(x_0)$ is the rational function field and $F_j=F_{j-1}(x_j)$ for every $j>0$. Denote by $m_j=[F_j:F_{j-1}]$ and consider $E=F_{i}$, for some index $i\geq 2$.  
Let $S$ be a set of places of $F_0$ that split completely in $E/F_0$ and such that  
$$\{Q \in \P(E): Q\cap F_0 \in S\}\cap\left(
\bigcup_{j=1}^{i}\{Q\in \P(E): \nu_Q(x_j)<0\}\right)=\emptyset$$ and let $\m{B}=\{Q \in \P(E): Q\cap F_0 \in S\}$.  Then, if $|S|=s>0$ we have that $|\m{B}|= sm$
where $m=m_{i}\ldots m_1$.  Choose a divisor $D$ of $F_0$ of degree $l$ such that $\supp(D)\cap S=\emptyset$, and let $D'$ be the conorm divisor of $D$ in $F_{i-1}$, i.e., $D'=\con_{F_{i-1}/F_0}(D).$ Let $\{f_1, \ldots, f_w \}$ be a basis for $\m{L}(D')$, the Riemann-Roch space associated to $D'$. Consider the space $V$ generated by 
$$\{f_j\,x_{i}^{e_{i}}:1\leq j \leq w \text{ and }0 \leq  e_i\leq m_i-2\}.$$ Since $\supp(D)$ is disjoint from S, the evaluation map 
$$
\begin{array}{cccl}
 ev: & V  &  \longrightarrow& \F_q^{ms} \\
 & f& \rightarrow  &  (f(P_{11}), \ldots, f(P_{ms}) )
\end{array}
$$ is well defined. The image of this mapping is a locally recoverable code with locality $r=m_{i}-1$, which we denote by
$C_{i}(S,D)$. The code coordinates are naturally partitioned into $\widetilde{m}s$ subsets of size $m_{i}$ each: \begin{align*}
A_t^j &= \{P_{tu}^j:
1 \leq u \leq m_{i}\}\\
& =\{Q\in \P(F_{i}): Q\cap F_{i-1}=P_{t1}^j\cap F_{i_0-1}=\widetilde{P_{t}}^j\}
\end{align*} where $1\leq j \leq s$; $1\leq t \leq \widetilde{m}$ and $\widetilde{m}=m_{i-1}\ldots m_1=m/m_{i}$. Denoting by $h=[F_{i}:\F_q(x_i)]$, we have that the parameters of the code satisfy

\begin{itemize}
\item $n= ms=m_{i}\ldots m_1s$,
 \item $k= w(m_{i}-1)\geq (\deg(D')+1-g(F_{i-1}))(m_{i}-1)=(\widetilde{m}l+1-g(F_{i-1}))(m_i-1)$,
 \item $d\geq n-lm-(m_{i}-2)h$,
\end{itemize}
provided that the right-hand side of the inequality for $d$ is a positive integer. Local recovery of an erased symbol $f (P_{tu}^{j})$ can be performed by polynomial interpolation through the positions of the points in the recovery set $A_t^j\setminus\{P_{tu}^j\}$.
\end{thm}

\begin{proof}
Straightforward from Theorem \ref{teobargtamo} in the particular  case of one step extension in the tower. 
\end{proof}

A modification of the previous result allow us to compute the exact dimension of the code, using a different vector space. Note that in this case, although we can compute the exact dimension, and improve the dimension in some cases, we have to pay a cost in the bound for the minimum distance of the code.

\begin{thm}\label{Teounr} Let $\m{F}=\{F_j\}_{j=0}^{\infty}$ be a sequence of function fields and $\{x_i\}_{j=0}^{\infty}$ a sequence of trascendental elements over $\F_q$ such that $F_0=\F_q(x_0)$ is the rational function field and $F_j=F_{j-1}(x_j)$ for every $j>0$. Denote by $m_j=[F_j:F_{j-1}]$ and consider $E=F_{i}$, for some index $i\geq 2$.  
Let $S$ be a set of places of $F_0$ that split completely in $E/F_0$ and such that  
$$\{Q \in \P(E): Q\cap F_0 \in S\}\cap\left(
\bigcup_{j=1}^{i}\{Q\in \P(E): \nu_Q(x_j)<0\}\right)=\emptyset$$ and let $\m{B}=\{Q \in \P(E): Q\cap F_0 \in S\}$.  Then, if $|S|=s>0$ we have that $|\m{B}|= sm$
where $m=m_{i}\ldots m_1$.  Choose a divisor $D$ of $F_0$ of degree $l$ such that $\supp(D)\cap S=\emptyset$, and let $\{f_1, \ldots, f_\ell \}$ a basis for $\m{L}(D)$, the Riemann-Roch space associated to $D$. Consider the space $V$ generated by 
$$\{f_wx_1^{e_1}\cdots x_{i}^{e_{i}}:1\leq w \leq \ell;\, 0\leq e_i \leq m_i-2\text{ and }0 \leq  e_j\leq m_j-1 \text{ for }j=1,\ldots, i-1 \}.$$ Since $\supp(D)$ is disjoint from S, the evaluation map 
$$
\begin{array}{cccl}
 ev: & V  &  \longrightarrow& \F_q^{ms} \\
 & f& \rightarrow  &  (f(P_{11}), \ldots, f(P_{ms}) )
\end{array}
$$ is well defined. The image of this mapping is a locally recoverable code $\m{C}$ with locality $m_{i}-1$, which we denote by
$C_{i}(S,D)$. The code coordinates are naturally partitioned into $\widetilde{m}s$ subsets of size $m_{i}$ each: \begin{align*}
A_t^j &= \{P_{tu}^j:
1 \leq u \leq m_{i}\}\\
& =\{Q\in \P(F_{i}): Q\cap F_{i-1}=P_{t1}^j\cap F_{i_0-1}=\widetilde{P_{t}}^j\}
\end{align*} where $1\leq j \leq s$; $1\leq t \leq \widetilde{m}$ and $\widetilde{m}=m_{i-1}\ldots m_1=m/m_{i}$. Denoting by $h_j=[F_{i}:\F_q(x_j)]$, we have that the parameters of the code satisfy

\begin{itemize}
 \item $n= ms=m_{i}\ldots m_1s$,
 \item $k= \ell(m_{i}-1)m_{i-1}\cdots m_1\geq (l+1)(m_{i}-1)m_{i-1}\cdots m_1$,
 \item $d\geq n-lm-(m_1-1)h_1-\cdots-(m_{i-1}-1)h_{i-1}-(m_{i}-2)h_{i}$.
\end{itemize}
provided that the right-hand side of the inequality for $d$ is a positive integer. Local recovery of an erased symbol $f (P_{tu}^{j})$ can be performed by polynomial interpolation through the positions of the points in the recovery set $A_t^j\setminus\{P_{tu}^j\}$.

\end{thm}

\begin{proof}
Since $[E:F_0]=\prod_{j=1}^{i}[F_j:F_{j-1}]=\prod_{j=1}^{i}m_{j}$ we have that $n=|\m{B}|=sm_{i}\ldots m_1$. The dimension of the code is the dimension of the vector space $V$ (since the evaluation map is injective). 
Note also that the bound of the minimum distance is just the length of a codeword minus the maximum number of zeroes that a function $f$ in $V$ can have.
In fact, $f_wx_1^{e_1}\cdots x_{i-1}^{e_{i-1}} \in F_{i-1}$, and there is a divisor $G$ of $F_{i-1}$ of degree $(l+(m_1-1)+\cdots+(m_{i-1}-1))\widetilde{m}$ such that $$f_wx_1^{e_1}\cdots x_{i-1}^{e_{i-1}} \in \m{L}(G)$$ and if $Q\in \supp(G)$ then $Q|P$ for some $P\in \supp(D)$ or $Q|Q_j$ where $Q_j$ is the simple pole of $x_j$ in $\F_q(x_j)$, for some $j=1, \ldots, i-1$. Therefore, if we denote by $\tilde{G}$ the conorm divisor of $G$ in $F_{i}$, i.e., $\tilde{G}=\con_{F_{i}/F_{i-1}}(G)$ and $\tilde{Q}_i=\con_{F_{i}/\F_q(x_i)}(Q_i)$ the conorm divisor of the simple pole $Q_i$ of $x_i$ in $\F_q(x_i)$, we have that 
 $$f_wx_1^{e_1}\cdots x_{i}^{e_{i}} \in \m{L}(\tilde{G}+(m_{i}-2)\tilde{Q}_{i})$$ where \begin{align*}
 \deg(\tilde{G})&=[F_{i}:F_{i-1}]\deg(G)\\
 &=[l+(m_1-1)+\cdots+(m_{i-1}-1)]\widetilde{m}\,m_i\\
 &=[l+(m_1-1)+\cdots+(m_{i-1}-1)]m,
 \end{align*} 
 and $\deg(\tilde{Q}_i)=[F_{i}:\F_q(x_i)]=h$. Finally, since a function $f\in V$ can have at most $$[l+(m_1-1)+\cdots+(m_{i-1}-1)]m+(m_{i}-2)h$$ poles (and thus zeros) we have a lower bound for the minimum distance. 
Now, let us show how the recovery of the coordinate is achieved. Let $f\in V$ and $$\boldsymbol{c}=(f(P_{11}^1), \ldots, 
f(P_{\widetilde{m}m_{i_0}}^j))$$ be a codeword. Assume, that the coordinate $f(P_{tu}^j)$ is missing, for some fixed $t$, $u$ and $j$. Recall that we denote by $P_j=P_{tu}^j\cap F_0$ and $\widetilde{P_t}^j=P_{tu}^j\cap F_{i-1}$ (See Figure \ref{figu0}). The recovery set for this coordinate is the evaluation of $f$ in the points of the set $A_t^j\setminus\{P_{tu}^j\}$, in other words, $\{f(P_{tk}^j)\}_{\underset{k\neq u}{1\leq k\leq r+1}}.$ Since $f_wx_1^{e_1}\cdots x_{i-1}^{e_{i-1}} \in F_{i-1}$ the evaluation of this in any point of the recovery set is constant and equal to the evaluation in $\widetilde{P_t}^j$. So, $f(P_{tu}^j)$ can be seen as the evaluation in $P_{tu}^j$ of a polynomial in  $T$: 
$$\tilde{f}(T):=\sum_{k=0}^{m_{i}-2}b_kT^k$$ with some appropriate chosen coefficients. And the same is true for all points in the recovery set (since they are all above the same place in $F_{i-1}$). Therefore, the coefficients $b_0, \ldots, b_{m_ {i}-1}$ can be found by polynomial interpolation of the remaining $m_{i}-1$ points in the recovery set.  
\end{proof}

 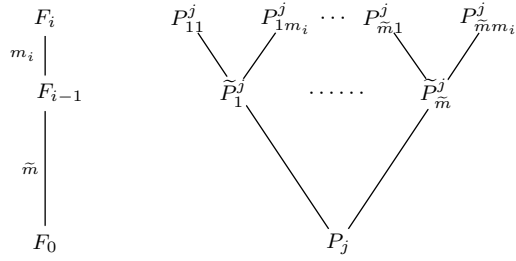
\begin{figure}[h!t]
        \begin{center}
  \begin{tikzpicture}[scale=1]
 \draw[line width=0.5 pt](-0.9,0)--(-0.9,3);
   \draw[white,  fill=white](-0.9,0) circle (0.23 cm);
   \draw[white,  fill=white](-0.9,2) circle (0.23 cm);
\draw[white,  fill=white](-0.9, 3) circle (0.23 cm);
\node at(-0.9,0){\footnotesize{$F_0$}};
\node at(-0.7,2){\footnotesize{$F_{i-1}$}};
\node at(-0.9,3){\footnotesize{$F_{i}$}};
     \draw[line width=0.5 pt](3,0)--(1,3);
     \draw[line width=0.5 pt](1.6,2)--(2.3,3);
                     \draw[line width=0.5 pt](3,0)--(5,3);
                     \draw[line width=0.5 pt](4.3,2)--(3.6,3);
        \draw[white,  fill=white](4,0) circle (0.23 cm);
  \draw[white,  fill=white](1.6,2) circle (0.23 cm);
    \draw[white,  fill=white](1,3) circle (0.23 cm);  
  \draw[white,  fill=white](2.3,3) circle (0.23 cm);    
 \draw[white,  fill=white](5,3) circle (0.23 cm);   
   \draw[white,  fill=white](3.6,3) circle (0.23 cm); 
     \draw[white,  fill=white](4.3,2) circle (0.23 cm);
 \draw[white,  fill=white](3,0) circle (0.23 cm);
   \node at(3,0){\footnotesize{$P_j$}};
 \node at(1.6,2){\footnotesize{$\widetilde{P}_1^j$}};
   \node at(3,2){\footnotesize{$\cdots\cdots$}};  
 \node at(4.3,2){\footnotesize{$\widetilde{P}_{\widetilde{m}}^j$}};         
  \node at(1,3){\footnotesize{$P_{11}^j$}};      
 \node at(2.3,3){\footnotesize{$P_{1m_{i}}^j$}};  
  \node at(2.95,3){\footnotesize{$\cdots$}};  
 \node at(3.6,3){\footnotesize{$P_{\widetilde{m}1}^j$}};      
    \node at(5,3){\footnotesize{$P_{\widetilde{m}m_{i}}^j$}};  
     \node at(-1.2,2.5){\tiny{$m_{i}$}};
          \node at(-1.1,1){\tiny{$\widetilde{m}$}};
 \end{tikzpicture}
  \caption{Diagram of an splitting place $P_j$ of $F_0$ in $F_{i}$}\label{figu0}
\end{center}\end{figure}

\begin{rem}
Notice that the elements in the set $S$ are the places that split completely in all the extensions $F_j/F_0$, $j=1, \ldots, i$ but this is not a necessary condition. Actually, the necessary condition is that the places in the set $S$ have above them places of degree one in $E$, and that we have enough of them to build a recovery set and bound the dimension of $V$ to have an injective application, but they could ramify in some intermediate extension. 
\end{rem}

\begin{example}\label{torreT}
For each $q>2$ let us consider the recursive tower $\mathcal{F}=\{F_j\}_{j=0}^{\infty}$ over $\F_{q^2}$ defined by Garcia and Stichtenoth in \cite{GS}, whose defining equation is
\begin{equation} \label{yq+y} 
y^q+y = \frac{x^q}{x^{q-1}+1}.
\end{equation}
It is known (see, for instance, \cite[Example 5.4.1]{NX}) that
$N(F_j) \geq q^{j}(q^2-q)+1$
and the genus is given by
$$ g(F_j) = \begin{cases} (q^{\frac{j+1}{2}}-1)^2 & \quad \text{for $j$ odd,} \\ 
(q^{\frac{j}{2}+1}-1)(q^{\frac{j}{2}}-1) & \quad \text{for $j$ even.}
\end{cases}$$
In this case, we have the set $$S=\{P_\alpha \in \P(F_0):\alpha \in \F_{q^2} \text{ and } \alpha^q+\alpha\neq 0\},$$ where $P_\alpha$ denotes the only single zero of $x_0-\alpha$ in the rational function field $F_0=\F_{q^2}(x_0)$. All the places in $S$ split completely in the tower, so for some $i>1$ we can consider $$\m{B}_{i}=\{Q\in \P(F_{i}):Q|P \text{ for some } P\in S\} $$ as the evaluation set. Notice that using the defining equation of the tower, is not hard to verify that the places in $\m{B}_{i}$ are not poles of $x_j$, $j=0, \ldots, i$. Moreover, the only pole of $x_0$ in $F_0$, that we denote by $P_\infty$ is totally ramified in the tower, and we can choose $D=lP_\infty$.  Since $\m{F}$ is an asymptotically good tower and every \textit{step} in the tower is of degree $q$, we have that $m=[F_{i}:F_0]=q^{i}=[F_i:\F_q^2(x_i)]=h$. Therefore, considering the set $$V_{i}=\langle x_0^{e_0}x_1^{e_1}\cdots x_{i}^{e_{i}}:0\leq e_0\leq l;\, 0\leq e_i \leq q-2\text{ and }0 \leq  e_j\leq q-1 \text{ for }j=1,\ldots, i-1\rangle,$$ we can build a locally recoverable code $C_{i}(S,D)$ with locality $q-1$, whose parameters satisfy  $$
\begin{array}{l}
 n= q^{i}(q^2-q)\\
 k= (l+1)(q-1)q^{i-1}\\
 d\geq (q^2-2q+2-l-(q-1)(i-1))q^i
\end{array}
$$ provided that the right-hand side of the third inequality for $d$ is a positive integer.

Notice that in the case  $i=1$ the same construction can be done, where $$V_1=\langle x_0^{e_0}x_1^{e_1}:0\leq e_0\leq l;\, 0\leq e_1 \leq q-2\rangle.$$
\end{example}

\begin{rem}
For $1 <i \leq q-1$ and $1 \leq l \leq (q-1)(q-i)$, the code $C_i(S,D)$ in the previous example verifies $d > 0$ since $$d \geq (q^2-2q+2-l-(q-1)(i-1))q^i = ((q-1)(q-i) - l+1)q^i > 0 $$ 
\end{rem}

\setcounter{example}{0}

\begin{example}(cont.)\label{torreTcont}
Taking into account the previous remark, if we want to maximize the dimension, we can choose $l=(q-1)(q-i)$ and in this case we obtain codes $C_{i}(S,D)$ for $i \leq q-1$ with the following relative parameters:

$$R=\frac{k}{n}= \frac{(l+1)}{q^2} \geq \frac{1}{q} \qquad \text{and} \qquad \delta =\frac{d}{n}\geq \frac{1}{q(q-1)}.$$


\end{example}

\section{A particular case for small steps}\label{sec4}

In this section we will work with the tower presented in Example~\ref{torreT}, and we will show that the bound presented in Theorem~\ref{Teounr} is sharp for the first two steps in the tower. 

\begin{pro}\label{exactmd}
	For $q\geq 5$, the code $C_2(S,D)$ of Example~\ref{torreT} with $D = qP_{\infty}$,
 is a locally recoverable code over $\F_{q^2}$, with locality $r=q-1$, whose parameters are \begin{align*}
	n&=q^2(q^2-q), \\
	k&= (q+1)q(q-1)=q^3-q, \\
	d&=q^2(q^2-q-q-(q-1)-(q-2)) = q^2(q^2-4q+3).\\
	\end{align*}
\end{pro}

\begin{proof}

The set $ S_0=\{\alpha \in \F_{q^2}: \alpha^q+\alpha\neq 0\},$ can be naturally partitioned into $q-1$ disjoint subsets \begin{equation}\label{Si} S_i=\left\{ \alpha \in \F_{q^2}: \dfrac{\alpha^{q}}{\alpha^{q-1}+1} =\dfrac{\alpha^{q+1}}{\alpha^{q}+\alpha} =: \beta_i \right\}, \quad \beta_i \in \F_{q}^*\end{equation} each one of size  $q$. 
For a rational place $P \in \P(F_j)$, we denote $$N(P) = N_{\F_{q^2}/ \F_{q}}( x_j(P)) = x_j(P)^{q+1}$$ and $$Tr(P) = Tr_{\F_{q^2}/ \F_{q}}( x_j(P)) = x_j(P)^q + x_j(P),$$ where $x_j$ is the trascendental element over $\F_{q^2}$ such that $F_j= F_{j-1}(x_j)$, for every step $j$ in the tower.
Using the defining equation of the tower, we have that if $Q$ is a place of $F_j$  over a place $P$ of $F_{j-1}$ such that $x_{j-1}(P) \in S_i$, then $$Tr(Q) = \dfrac{N(P)}{Tr(P)}= \beta_i.$$ %
In particular we obtain that, for each place $P \in \mathcal{B}\cap \P(F_{j-1})$ and each $S_i$, there are at most two places $Q_1$ and $Q_2$ of $\mathcal{B}$ such that $$x_{j-1}(Q_1)=x_{j-1}(Q_2)=x_{j-1}(P) \in S_i,$$ since there are at most two elements in $\F_{q^2}$ with the same norm and trace. Moreover, if $\sigma$ is the only non trivial automorphism in $\F_{q^2}/\F_{q}$, we have that $$x_j(Q_1)=\sigma(x_j(Q_2)),$$ and since $q$ is odd, there is exactly one remaining place $Q|P$ such that $$x_{j-1}(Q)=x_{j-1}(P) \in S_i,\quad \text{and}\quad\sigma(x_j(Q))=x_j(Q).$$ In other words, if we consider the $q-1$ disjoint sets \begin{equation}\label{Bi} B_i=\{\alpha \in \F_{q^2}: Tr(\alpha) = \beta_i\}, \qquad \text{for } \beta_i \in \F_{q}^*,\end{equation} 
we have that
for any indexes $1\leq i,k \leq q-1$, \begin{equation}\label{equ1}
|S_k\cap B_i|=|\{Q\in \P(F_j):x_{j}(Q) \in S_k \text{ and }x_j(Q)\in B_i\}| \leq 2\end{equation} and \begin{equation}\label{equ2}|\{Q\in \P(F_j):x_{j}(Q) \in S_k \text{ and }x_j(Q)\in B_i\}| =1 \,\, \text{ if and only if }\,\, x_j(Q)\in \F_q.\end{equation}

Now,  using the notation in Theorem~\ref{Teounr}, we take $i=2$ and $l=q$ and we consider $D = lP_{\infty}$ so that $V$ is generated by $$\{x_0^{e_0}x_1^{e_1}x_2^{e_2} : 0 \leq e_0 \leq q; 0 \leq e_1 \leq q-1; 0 \leq e_2 \leq q-2 \}.$$ From the proof of the theorem we have that any function on $V$ has at most $$(l+1)q^2+q^3 + (q-1)q^2=q^2(3q-3)$$ zeros, and we will now exhibit a function that actually achieves that number of zeros.
Consider an element $\beta_1 \in \F_{q^2}^*$ and set $$h_0 = \prod_{\alpha \in S_1}\left( x_0-\alpha \right),$$ so we have exactly $q^3=|S_1|[F_2:F_0]$ places $Q \in \m{B}$ that are zeroes of $h_0$. In other words, for any of these places we have that $x_0(Q) \in S_i$ and  $x_1(Q) \in B_1$, by the recursive definition of the tower. 
Now since $q$ is odd, chose $\beta_i$ such that $$|S_i\cap B_1|= 1$$  and set $$h_1 = \prod_{\alpha \in S_i  \setminus B_1} \left( x_1-\alpha \right),$$ so we have $(q-1)q^2=(|S_i|-1)[F_2:F_1]|\{Q\in F_1: x_1(Q)=\alpha\}|$ new places $Q \in \m{B}$ that are zeroes of $h_1$. 

Since $|B_1|=q$, from \eqref{equ1} and \eqref{equ2}, we can write $B_1$ as a disjoint union of $(q+1)/2$ non empty sets $$B_1=\bigcup_{k=1}^{\frac{q+1}{2}}B_1\cap S_{i_k}$$ and therefore if $Q\in \mathcal{B}$ is such that $x_0(Q)\in S_1$ or $x_1(Q)=S_{i}$ then $x_2(Q)$ is one of the $q(q+1)/2$ different values in $\F_{q^2}$. 
Now, since $q\geq 5$ implies that $q \left( \frac{q+1}{2} \right)< q^2-q-(q-2)$ we can choose $q-2$ values $\gamma_1, \ldots, \gamma_{q-2}$ in $S_0$ such that $x_2(Q)=\gamma_i$ for $Q\in F_2$ and $$h_2 = \prod_{i=1}^{q-2}  \left( x_2-\gamma_i \right),$$ has $q^2(q-2)$ new different zeros in $\mathcal{B}$. 

Therefore, $f = h_0h_1h_2 \in V$ has exactly $q^3+(q-1)q^2 + (q-2)q^2 = q^2(3q-3)$ zeroes, attaining the bound provided by Theorem~\ref{Teounr}. 
\end{proof}

 \begin{figure}[h!t]
	\begin{center}
		\begin{tikzpicture}[scale = 1.5]
			\clip(-1.4,-0.2) rectangle (5,3.4);
			\draw [line width=2pt,color=wrwrwr] (-0.8,2)-- (0,0);
			\draw [line width=2pt,color=wrwrwr] (-0.4,2)-- (0,0);
			\draw [line width=2pt,color=wrwrwr] (0,1.8)-- (0,0);
			\draw [line width=2pt,color=wrwrwr] (0.4,1.8)-- (0,0);
			\draw [line width=2pt,color=wrwrwr] (0.8,1.6)-- (0,0);
			\draw [line width=2pt,color=wrwrwr] (1.2,2)-- (2,0);
			\draw [line width=2pt,color=wrwrwr] (1.6,2)-- (2,0);
			\draw [line width=2pt,color=wrwrwr] (2,1.8)-- (2,0);
			\draw [line width=2pt,color=wrwrwr] (2.4,1.8)-- (2,0);
			\draw [line width=2pt,color=wrwrwr] (2.8,1.6)-- (2,0);
			\draw [line width=2pt,color=wrwrwr] (3.2,2)-- (4,0);
			\draw [line width=2pt,color=wrwrwr] (3.6,2)-- (4,0);
			\draw [line width=2pt,color=wrwrwr] (4,1.8)-- (4,0);
			\draw [line width=2pt,color=wrwrwr] (4.4,1.8)-- (4,0);
			\draw [line width=2pt,color=wrwrwr] (4.8,1.6)-- (4,0);
			\draw [line width=2pt,color=wrwrwr] (-1,2.4)-- (-0.8,2);
			\draw [line width=2pt,color=wrwrwr] (-1,2.6)-- (-0.8,2);
			\draw [line width=2pt,color=wrwrwr] (-1,2.8)-- (-0.8,2);
			\draw [line width=2pt,color=wrwrwr] (-1,3)-- (-0.8,2);
			\draw [line width=2pt,color=wrwrwr] (-1,3.2)-- (-0.8,2);
			\draw [line width=2pt,color=wrwrwr] (-0.6,2.4)-- (-0.4,2);
			\draw [line width=2pt,color=wrwrwr] (-0.6,2.6)-- (-0.4,2);
			\draw [line width=2pt,color=wrwrwr] (-0.6,3)-- (-0.4,2);
			\draw [line width=2pt,color=wrwrwr] (-0.6,3.2)-- (-0.4,2);
			\draw [line width=2pt,color=wrwrwr] (-0.6,2.8)-- (-0.4,2);
			\draw [line width=2pt,color=wrwrwr] (-0.2,2.4)-- (0,1.8);
			\draw [line width=2pt,color=wrwrwr] (-0.2,2.6)-- (0,1.8);
			\draw [line width=2pt,color=wrwrwr] (-0.2,2.8)-- (0,1.8);
			\draw [line width=2pt,color=wrwrwr] (-0.2,3)-- (0,1.8);
			\draw [line width=2pt,color=wrwrwr] (-0.2,3.2)-- (0,1.8);
			\draw [line width=2pt,color=wrwrwr] (0.2,2.4)-- (0.4,1.8);
			\draw [line width=2pt,color=wrwrwr] (0.2,2.6)-- (0.4,1.8);
			\draw [line width=2pt,color=wrwrwr] (0.2,2.8)-- (0.4,1.8);
			\draw [line width=2pt,color=wrwrwr] (0.2,3)-- (0.4,1.8);
			\draw [line width=2pt,color=wrwrwr] (0.2,3.2)-- (0.4,1.8);
			\draw [line width=2pt,color=wrwrwr] (0.6,2.4)-- (0.8,1.6);
			\draw [line width=2pt,color=wrwrwr] (0.6,2.6)-- (0.8,1.6);
			\draw [line width=2pt,color=wrwrwr] (0.6,3)-- (0.8,1.6);
			\draw [line width=2pt,color=wrwrwr] (0.6,2.8)-- (0.8,1.6);
			\draw [line width=2pt,color=wrwrwr] (0.6,3.2)-- (0.8,1.6);
			\draw [line width=2pt,color=wrwrwr] (1,2.4)-- (1.2,2);
			\draw [line width=2pt,color=wrwrwr] (1,2.6)-- (1.2,2);
			\draw [line width=2pt,color=wrwrwr] (1,2.8)-- (1.2,2);
			\draw [line width=2pt,color=wrwrwr] (1,3)-- (1.2,2);
			\draw [line width=2pt,color=wrwrwr] (1,3.2)-- (1.2,2);
			\draw [line width=2pt,color=wrwrwr] (1.4,2.4)-- (1.6,2);
			\draw [line width=2pt,color=wrwrwr] (1.4,2.6)-- (1.6,2);
			\draw [line width=2pt,color=wrwrwr] (1.4,3)-- (1.6,2);
			\draw [line width=2pt,color=wrwrwr] (1.4,3.2)-- (1.6,2);
			\draw [line width=2pt,color=wrwrwr] (1.4,2.8)-- (1.6,2);
			\draw [line width=2pt,color=wrwrwr] (1.8,2.4)-- (2,1.8);
			\draw [line width=2pt,color=wrwrwr] (1.8,2.6)-- (2,1.8);
			\draw [line width=2pt,color=wrwrwr] (1.8,2.8)-- (2,1.8);
			\draw [line width=2pt,color=wrwrwr] (1.8,3)-- (2,1.8);
			\draw [line width=2pt,color=wrwrwr] (1.8,3.2)-- (2,1.8);
			\draw [line width=2pt,color=wrwrwr] (2.2,2.4)-- (2.4,1.8);
			\draw [line width=2pt,color=wrwrwr] (2.2,2.6)-- (2.4,1.8);
			\draw [line width=2pt,color=wrwrwr] (2.2,2.8)-- (2.4,1.8);
			\draw [line width=2pt,color=wrwrwr] (2.2,3)-- (2.4,1.8);
			\draw [line width=2pt,color=wrwrwr] (2.2,3.2)-- (2.4,1.8);
			\draw [line width=2pt,color=wrwrwr] (2.6,2.4)-- (2.8,1.6);
			\draw [line width=2pt,color=wrwrwr] (2.6,2.6)-- (2.8,1.6);
			\draw [line width=2pt,color=wrwrwr] (2.6,3)-- (2.8,1.6);
			\draw [line width=2pt,color=wrwrwr] (2.6,2.8)-- (2.8,1.6);
			\draw [line width=2pt,color=wrwrwr] (2.6,3.2)-- (2.8,1.6);
			\draw [line width=2pt,color=wrwrwr] (3,2.4)-- (3.2,2);
			\draw [line width=2pt,color=wrwrwr] (3,2.6)-- (3.2,2);
			\draw [line width=2pt,color=wrwrwr] (3,2.8)-- (3.2,2);
			\draw [line width=2pt,color=wrwrwr] (3,3)-- (3.2,2);
			\draw [line width=2pt,color=wrwrwr] (3,3.2)-- (3.2,2);
			\draw [line width=2pt,color=wrwrwr] (3.4,2.4)-- (3.6,2);
			\draw [line width=2pt,color=wrwrwr] (3.4,2.6)-- (3.6,2);
			\draw [line width=2pt,color=wrwrwr] (3.4,2.8)-- (3.6,2);
			\draw [line width=2pt,color=wrwrwr] (3.4,3)-- (3.6,2);
			\draw [line width=2pt,color=wrwrwr] (3.4,3.2)-- (3.6,2);
			\draw [line width=2pt,color=wrwrwr] (3.8,2.4)-- (4,1.8);
			\draw [line width=2pt,color=wrwrwr] (3.8,2.6)-- (4,1.8);
			\draw [line width=2pt,color=wrwrwr] (3.8,2.8)-- (4,1.8);
			\draw [line width=2pt,color=wrwrwr] (3.8,3)-- (4,1.8);
			\draw [line width=2pt,color=wrwrwr] (3.8,3.2)-- (4,1.8);
			\draw [line width=2pt,color=wrwrwr] (4.2,2.4)-- (4.4,1.8);
			\draw [line width=2pt,color=wrwrwr] (4.2,2.6)-- (4.4,1.8);
			\draw [line width=2pt,color=wrwrwr] (4.2,2.8)-- (4.4,1.8);
			\draw [line width=2pt,color=wrwrwr] (4.2,3)-- (4.4,1.8);
			\draw [line width=2pt,color=wrwrwr] (4.2,3.2)-- (4.4,1.8);
			\draw [line width=2pt,color=wrwrwr] (4.8,1.6)-- (4.6,2.4);
			\draw [line width=2pt,color=wrwrwr] (4.6,2.6)-- (4.8,1.6);
			\draw [line width=2pt,color=wrwrwr] (4.8,1.6)-- (4.6,2.8);
			\draw [line width=2pt,color=wrwrwr] (4.6,3)-- (4.8,1.6);
			\draw [line width=2pt,color=wrwrwr] (4.8,1.6)-- (4.6,3.2);
			\begin{scriptsize}
				\draw [fill=rvwvcq] (2,0) circle (2.5pt);
				\draw [fill=sexdts] (4,0) circle (2.5pt);
				\draw [fill=rvwvcq] (0,0) circle (2.5pt);
				\draw [fill=rvwvcq] (-0.8,2) circle (2.5pt);
				\draw [fill=rvwvcq] (-0.4,2) circle (2.5pt);
				\draw [fill=sexdts] (0,1.8) circle (2.5pt);
				\draw [fill=sexdts] (0.4,1.8) circle (2.5pt);
				\draw [fill=dtsfsf] (0.8,1.6) circle (2.5pt);
				\draw [fill=sexdts] (2,1.8) circle (2.5pt);
				\draw [fill=rvwvcq] (1.6,2) circle (2.5pt);
				\draw [fill=rvwvcq] (1.2,2) circle (2.5pt);
				\draw [fill=sexdts] (2.4,1.8) circle (2.5pt);
				\draw [fill=dtsfsf] (2.8,1.6) circle (2.5pt);
				\draw [fill=sexdts] (3.2,2) circle (2.5pt);
				\draw [fill=sexdts] (3.6,2) circle (2.5pt);
				\draw [fill=dtsfsf] (4,1.8) circle (2.5pt);
				\draw [fill=dtsfsf] (4.4,1.8) circle (2.5pt);
				\draw [fill=black] (4.8,1.6) circle (2.5pt);
				\draw [fill=sexdts] (-1,2.6) circle (2.5pt);
				\draw [fill=sexdts] (-1,2.8) circle (2.5pt);
				\draw [fill=rvwvcq] (-1,3) circle (2.5pt);
				\draw [fill=rvwvcq] (-1,3.2) circle (2.5pt);
				\draw [fill=dtsfsf] (-1,2.4) circle (2.5pt);
				\draw [fill=dtsfsf] (-0.6,2.4) circle (2.5pt);
				\draw [fill=sexdts] (-0.6,2.6) circle (2.5pt);
				\draw [fill=sexdts] (-0.6,2.8) circle (2.5pt);
				\draw [fill=rvwvcq] (-0.6,3) circle (2.5pt);
				\draw [fill=rvwvcq] (-0.6,3.2) circle (2.5pt);
				\draw [fill=black] (-0.2,2.4) circle (2.5pt);
				\draw [fill=dtsfsf] (-0.2,2.6) circle (2.5pt);
				\draw [fill=dtsfsf] (-0.2,2.8) circle (2.5pt);
				\draw [fill=sexdts] (-0.2,3) circle (2.5pt);
				\draw [fill=sexdts] (-0.2,3.2) circle (2.5pt);
				\draw [fill=black] (0.2,2.4) circle (2.5pt);
				\draw [fill=dtsfsf] (0.2,2.6) circle (2.5pt);
				\draw [fill=dtsfsf] (0.2,2.8) circle (2.5pt);
				\draw [fill=sexdts] (0.2,3) circle (2.5pt);
				\draw [fill=sexdts] (0.2,3.2) circle (2.5pt);
				\draw [fill=rvwvcq] (0.6,2.4) circle (2.5pt);
				\draw [fill=black] (0.6,2.6) circle (2.5pt);
				\draw [fill=black] (0.6,2.8) circle (2.5pt);
				\draw [fill=dtsfsf] (0.6,3) circle (2.5pt);
				\draw [fill=dtsfsf] (0.6,3.2) circle (2.5pt);
				\draw [fill=dtsfsf] (1,2.4) circle (2.5pt);
				\draw [fill=sexdts] (1,2.6) circle (2.5pt);
				\draw [fill=sexdts] (1,2.8) circle (2.5pt);
				\draw [fill=rvwvcq] (1,3) circle (2.5pt);
				\draw [fill=rvwvcq] (1,3.2) circle (2.5pt);
				\draw [fill=dtsfsf] (1.4,2.4) circle (2.5pt);
				\draw [fill=sexdts] (1.4,2.6) circle (2.5pt);
				\draw [fill=rvwvcq] (1.4,3) circle (2.5pt);
				\draw [fill=rvwvcq] (1.4,3.2) circle (2.5pt);
				\draw [fill=sexdts] (1.4,2.8) circle (2.5pt);
				\draw [fill=black] (1.8,2.4) circle (2.5pt);
				\draw [fill=dtsfsf] (1.8,2.6) circle (2.5pt);
				\draw [fill=dtsfsf] (1.8,2.8) circle (2.5pt);
				\draw [fill=sexdts] (1.8,3) circle (2.5pt);
				\draw [fill=sexdts] (1.8,3.2) circle (2.5pt);
				\draw [fill=black] (2.2,2.4) circle (2.5pt);
				\draw [fill=dtsfsf] (2.2,2.6) circle (2.5pt);
				\draw [fill=dtsfsf] (2.2,2.8) circle (2.5pt);
				\draw [fill=sexdts] (2.2,3) circle (2.5pt);
				\draw [fill=sexdts] (2.2,3.2) circle (2.5pt);
				\draw [fill=rvwvcq] (2.6,2.4) circle (2.5pt);
				\draw [fill=black] (2.6,2.6) circle (2.5pt);
				\draw [fill=dtsfsf] (2.6,3) circle (2.5pt);
				\draw [fill=black] (2.6,2.8) circle (2.5pt);
				\draw [fill=dtsfsf] (2.6,3.2) circle (2.5pt);
				\draw [fill=black] (3,2.4) circle (2.5pt); 
				\draw [fill=dtsfsf] (3,2.6) circle (2.5pt);
				\draw [fill=dtsfsf] (3,2.8) circle (2.5pt);
				\draw [fill=sexdts] (3,3) circle (2.5pt);
				\draw [fill=sexdts] (3,3.2) circle (2.5pt);
				\draw [fill=black] (3.4,2.4) circle (2.5pt);
				\draw [fill=dtsfsf] (3.4,2.6) circle (2.5pt);
				\draw [fill=dtsfsf] (3.4,2.8) circle (2.5pt);
				\draw [fill=sexdts] (3.4,3) circle (2.5pt);
				\draw [fill=sexdts] (3.4,3.2) circle (2.5pt);
				\draw [fill=rvwvcq] (3.8,2.4) circle (2.5pt);
				\draw [fill=black] (3.8,2.6) circle (2.5pt);
				\draw [fill=black] (3.8,2.8) circle (2.5pt);
				\draw [fill=dtsfsf] (3.8,3) circle (2.5pt);
				\draw [fill=dtsfsf] (3.8,3.2) circle (2.5pt);
				\draw [fill=rvwvcq] (4.2,2.4) circle (2.5pt);
				\draw [fill=black] (4.2,2.6) circle (2.5pt);
				\draw [fill=black] (4.2,2.8) circle (2.5pt);
				\draw [fill=dtsfsf] (4.2,3) circle (2.5pt);
				\draw [fill=dtsfsf] (4.2,3.2) circle (2.5pt);
				\draw [fill=sexdts] (4.6,2.4) circle (2.5pt);
				\draw [fill=rvwvcq] (4.6,2.6) circle (2.5pt);
				\draw [fill=rvwvcq] (4.6,2.8) circle (2.5pt);
				\draw [fill=black] (4.6,3) circle (2.5pt);
				\draw [fill=black] (4.6,3.2) circle (2.5pt);
			\end{scriptsize}
		\end{tikzpicture}
		\caption{Diagram of three splitting places $P_j$ of $F_0$ in $F_2/F_0$, for $q=5$. Each colour, in each function field, represent a set $S_i$, for $1 \leq i \leq 4=q-1$}\label{figuplaces}
\end{center}\end{figure}
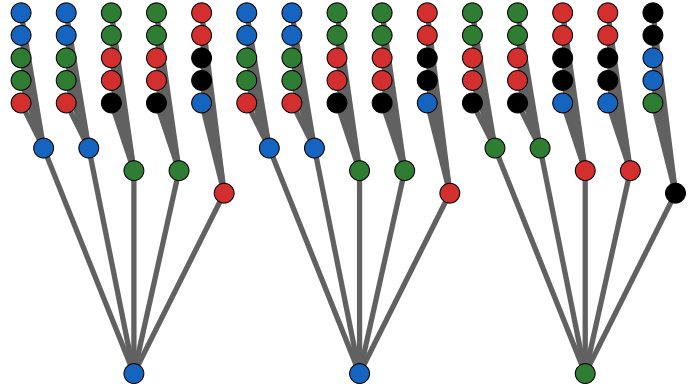

\begin{pro}\label{prop42}
	For each odd $q\geq 5$, the code $C_2(S,D)$ in Example~\ref{torreT} with $D = \frac{q(q-1)}{2}P_{\infty}$,
 is a locally recoverable code over $\F_{q^2}$, with locality $r=q-1$, whose parameters are 
  \begin{align*}
	n&=q^2(q^2-q), \\
	k&=q(q-1) \left(q \left(\frac{q-1}{2}-1 \right) +1\right), \\
	d&=\frac{1}{2}q^2\left(q^2 -3q +6 \right).\\
	\end{align*}

\end{pro}

\begin{proof}
	We follow the same notation as in the previous proposition and we only consider two steps in the tower. 
	Chose one subset, let us say $S_1$, which we know it has $q$ elements.  So, $$|\{P \in \P(F_1):x_1(P)\in S_1\}|=q^2$$ since for each value $x_1(P) \in S_1$ we have $q$ different places $P_\alpha \in \P(F_0)$ under $P$. Since  these places in $F_1$ 
	come ``in pairs" (with the only exceptions of places $P$ with $x_1(P) \in \F_q$) and there are exactly $q$ places $P_{\alpha} \in \P(F_0)$ with only one place $P|P_{\alpha}$ and $x_1(P)$ $ \in S_1$, 
	we have that there are exactly  $ \frac{q^2-q}{2} = q \frac{q-1}{2}$ places $P_{\alpha} \in \P(F_0)$ with two distinct places $P|P_{\alpha}$, $\sigma(P)|P_{\alpha}$ and $x_1(P) \in S_1$. Now we choose the remaining $q^2-q-q\frac{q-1}{2}-q = q\frac{q-1}{2}-q$ places $P_{\alpha} \in S\subset \P(F_0)$ such that for all $P\in \P(F_1)$ with $P| P_{\alpha}$ we have that $x_1(P) \not\in S_1$. Set then
	$$\mathbb{H}_0 = \{P_\alpha \in S: \: x_1(P) \not\in S_1  \text{ for all }P | P_{\alpha}\},$$ of cardinality $q \left(\frac{q-1}{2}-1\right)$, and put $H_0=\{x_0(P): P \in \mathbb{H}_0\}$.
	
	Now, defining $$h_0 = \prod_{s \in H_0} (x_0 -s) \in V$$ we have that $h_0$ has $q^3 \left(\frac{q-1}{2}-1\right) = \frac{q^4-3q^3}{2}$ different zeros in $\m{B}$. 
	
	Now, among all the $q^2(q^2-q)$ places $Q$ in $\m{B} \subset \P(F_2)$, such that $Q\cap F_0 \in S\setminus \mathbb{H}_0 $, the value $x_1(Q)$ can attain $q^2-q$ possible values, $q$ of those are in $S_1$. Since $$|\{x_1(Q):Q\in \m{B} \text{ and }Q\cap F_0 \in \mathbb{H}_0 \}|=q \left(\frac{q-1}{2}-1\right)$$ (recall that the places are naturally partitioned into subsets of size $q$ and inside each subset these places decompose in the same way) and $q \geq 5$, then we have $$(q^2-q)-q - q\left(\frac{q-1}{2}-1\right) \geq q,$$ so we can build a non empty set $H_1\subset \F_{q^2}$ of cardinality $q-1$ such that $$ H_1\subset \{x_1(Q): Q\in F_1, Q\cap F_0 \not\in H_0 \text{ and } x_1(Q) \not\in S_1\}  
	$$ 
	and  this allow us to form a function $$h_1 = \prod_{s \in H_1} (x_1 - s) \in V$$  of degree $q-1$ with $q^2(q-1)$ zeroes in $\m{B}$.
	
	Finally, since no point $Q$ that is a zero of $h_0$ or $h_1$ verifies $x_1(Q) \in S_1$ by construction of $H_0$ and $H_1$, we can consider some subset $H_2 \subset \{ x_2(Q): \: x_1(Q) \in S_1 \} = B_1 \subset \F_{q^2}$ of size $q-2$ and set $$h_2 = \prod_{s \in H_2} (x_2 - s)$$ of degree $q-2$ with $q^2(q-2)$ distinct zeroes in $\m{B}$.
	
	Since the subsets $\{Q \in \m{B}: x_0(Q) \in H_0 \}, \{Q \in \m{B}: x_1(Q) \in H_1 \}, \{Q \in \m{B}: x_2(Q) \in H_2 \}$ are pairwise disjoints by construction of $H_0,H_1,H_2$, we obtain that $f = h_0h_1h_2 \in V$ has exactly $\frac{q^4-3q^3}{2}+(q-1)q^2 +q^2(q-2) =  \frac{1}{2} q^2 \left( q^2 + q  -6 \right)$ zeroes, and thus the minimum distance is exactly \begin{align*} q^2(q^2-q) - \frac{1}{2} q^2 \left( q^2 + q  -6 \right) & = \frac{1}{2}q^2 (q^2-3q+6) \\ & = q^2 \left(q^2 -2q+2 -q  \left( \frac{q-1}{2} -1 \right) - (q-1) \right)\end{align*} 
	attaining the bound provided by the theorem and the example. 
	
\end{proof}

The above results show that, in the considered tower, the lower bound for the minimum distance was attained for a small step in the tower, that is, considering the extension $F_2$ of $F_0$ over $\F_{q^2}$. Next, we show that in further steps (i.e., considering extensions $F_j/F_i$ for $j \geq i+3$) there are common zeroes of $x_j$ and $x_i$, so the bound will never be attained for this tower.

In the following lemma, we use the same notation as in the previous propositions.

\begin{lem}
	Consider the field extensions $F_j/F_0$ over $\F_{q^2}$ of the Garcia-Stichtenoth tower of Example \ref{torreT}. Set $\m{B} \subset \P(F_j)$ the evaluation set with $q^j(q^2-q)$ places and $S = \{Q \cap F_0, \: Q \in \m{B}\} \subset \P(F_0)$. For $S_0 = \bigcup_{j=1}^{q-1} S_i \subset \F_{q^2}$, we have 
	\begin{enumerate}
		\item For $0 \leq i \leq j$ and $\alpha \in \F_{q^2} \setminus S_0$, the function $x_i - \alpha$ have no zeroes in $\m{B}$. For $\alpha \in S_0$, the function $x_i - \alpha$ have exactly $q^j$ zeroes in $\m{B}$.
		\item For $1 \leq i \leq j$ and $\alpha \in S_0$, if $Q \in \P(F_i)$ is a zero of $x_i-\alpha$, then $P=Q \cap F_{i-1}$ verifies $x_{i-1}(P) \in S_k$, where $\alpha \in B_k$,
		\item Conversely, for $1 \leq i \leq j$ and $\alpha \in S_0$, if $P \in \P(F_{i-1})$, verifies $x_{i-1}(P) \in S_k$, where $\alpha \in B_k$, then exists exactly one place $Q \in \P(F_i)$ with $Q|P$ that is a zero of $x_i-\alpha$.
	\end{enumerate}
\end{lem}

\begin{proof}
The proof of the first part follows from the definition of $S_0$ and the fact that each place $P \in \P(F_{i-1})\cap \m{B}$ has exactly $q$ places of $F_i$ above it. The second and third part are straightforward from the definitions (see equations \eqref{Si} and \eqref{Bi}) of the sets $S_i$ and $B_i$, for $1 \leq i \leq j$ and the recursive defining equation of the tower.
\end{proof}

Since there are exactly $q^j$ places $Q \in \P(F_j)$ laying above a place in $S=\m{B}\cap\P(F_0)$, and exactly $q^{j-1}(q^2-q)=q^j(q-1)$ places in $\P(F_{j-1})$, we obtain a one-to-one bijection between zeroes $Q$ of $x_j - \alpha$ and places $P \in \P(F_{j-1})$ with $x_{j-1}(P) \in S_k$, where $\alpha \in B_k$. 

\begin{defin}
	We say $Q \in \P(F_j)$ \textit{it is colored} $k$ if $x_j(Q) \in S_k$ for some $1 \leq k \leq q-1$, where $Q$ lays under some place in $\m{B}$.
\end{defin}

Let $i \geq 0$, and $j \geq i+3$. There are exactly $\frac{q+1}{2}$ different colours appearing over some place $P \in \P(F_i)$. On the other hand, when considering places $F_j$, each colour appears above exactly $\frac{q+1}{2}$ different colours; that is, for fixed $S_l$ there are exactly $\frac{q+1}{2}$ indexes $\{l_1, \ldots, l_{ \frac{q+1}{2} } \}$ such that for $Q \in \P(F_j)$ with $x_j(Q) \in S_l$ it is $x_{j-1}(Q \cap F_{j-1}) \in S_{k}$ for some $k \in \{l_1, \ldots, l_{ \frac{q+1}{2} } \}$.

Recall that two places of the same colour, when decompose, they decompose each in $q$ places of the same $\frac{q+1}{2}$ colours, independently of the function field they belong.

\begin{cor}\label{cerosencomun}
	For $j \geq i+3$ and any $\alpha, \beta \in S_0$, the functions $x_i-\alpha$ and $x_j-\beta$ have at least one common zero $Q \in \m{B}$.
\end{cor}

\begin{proof} Let $Q_0 \in \m{B}$ be a zero of $x_j-\beta$ and $P \in \P(F_i)$ be a zero of $x_i-\alpha$. Without loss of generality, we can assume $P$ is of colour \textit{blue} and $Q_0 \cap F_{i+2}$ is \textit{black}. By the pigeonhole principle, among the $\frac{q+1}{2}$ colours above $P$ and among the $\frac{q+1}{2}$ colours below any black place, there is at least one colour in common, let us say \textit{green}. Then, there is a place $P_0 \in \P(F_{i+1})$ such that $P_0|P$ and $P_0$ is green, therefore there is also a place $P_1 \in \P(F_{i+2})$ such that $P_1|P_0$ and $P_1$ is black.  Finally, since $P_1$ is black, there is also a place $Q \in \P(F_j)$ such that $Q|P_1$ and $x_j(Q) = x_j(Q_0) = \beta$. Since $Q | P$ we also have $x_i(Q)=x_i(P) = \alpha$. Thus, $Q \in \m{B}$ and it is a common zero of $x_i-\alpha$ and $x_j-\beta$.
\end{proof}

 \begin{figure}[h!t]
        \begin{center}
  \begin{tikzpicture}[scale=0.95]
 \draw[line width=0.5 pt](-0.9,0)--(-0.9,3.5);
   \draw[white,  fill=white](-0.9,0) circle (0.23 cm);
      \draw[white,  fill=white](-0.9,1) circle (0.23 cm);
   \draw[white,  fill=white](-0.9,2) circle (0.23 cm);
\draw[white,  fill=white](-0.9, 3.5) circle (0.23 cm);
\node at(-0.9,0){\footnotesize{$F_i$}};
\node at(-0.9,1){\footnotesize{$F_{i+1}$}};
\node at(-0.9,2){\footnotesize{$F_{i+2}$}};
\node at(-0.9,3.5){\footnotesize{$F_{j}$}};

 \draw[line width=0.5 pt](1.9,2)--(1.9,3.5);
   \draw[white,  fill=white](1.9,2) circle (0.23 cm);
\draw[white,  fill=white](1.9, 3.5) circle (0.23 cm);

\node at(1.9,2){\footnotesize{$Q_0\cap F_{i+2}$}};
\node at(1.9,3.5){\footnotesize{$Q_0$}};

 \draw[line width=0.5 pt](3.9,0)--(3.9,3.5);
   \draw[white,  fill=white](3.9,0) circle (0.23 cm);
      \draw[white,  fill=white](3.9,1) circle (0.23 cm);
   \draw[white,  fill=white](3.9,2) circle (0.23 cm);
\draw[white,  fill=white](3.9, 3.5) circle (0.23 cm);
\node at(3.9,0){\footnotesize{$P$}};
\node at(3.9,1){\footnotesize{$P_0$}};
\node at(3.9,2){\footnotesize{$P_1$}};
\node at(3.9,3.5){\footnotesize{$Q$}};

 \end{tikzpicture}
  \caption{Diagram of extension and places of Corollary \ref{cerosencomun}}\label{figu0bis}
\end{center}\end{figure}
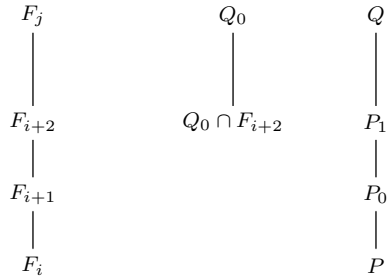

The previous corollary, make us conjecture the following:

\begin{conjecture}
Any function $f \in  V$ will have fewer than $$lq^{k}+(k-1)(q-1)q^{k-1}+(q-2)q^{k-1}$$ zeroes in $\m{B}$, and therefore the lower bound for the distance in Theorem \ref{Teounr} will not be attained for $i \geq 3$.
\end{conjecture}

\section{Relative parameters and asymptotics}

In this last section, we will analyze the relative parameters $R$ and $\delta$ of the codes $C_i(S,D)$ built in the previous sections, using the tower described in Example \ref{torreT}.

%

\begin{pro}\label{propcotai} For each odd $q \geq 5$ and any $2 \leq i \leq q-1$, $1 \leq l\leq  (q-1)(q-i)$, the relative parameters $R$ and $\delta$ of $C_i(S,D)$ with $D = l P_{\infty}$ verify:
	
$$ R + \frac{q-1}{q} \delta > \frac{r}{r+1} \left( \frac{q-i}{q} \right),$$  where $r = q-1$ is the locality of $C_i$.
\end{pro}

\begin{proof}
	Since $\delta  \geq \dfrac{((q-1)(q-i) - l+1)q^i}{q^i(q^2-q)} = \dfrac{q}{q-1} \dfrac{(q-1)(q-i)-l+1}{q^2}$ we have $$R + \frac{q-1}{q} \delta \geq \frac{(q-1)(q-i)+2}{q^2} > \frac{(q-1)(q-i)}{q^2} = \frac{r}{r+1} \left( \frac{q-i}{q} \right).$$ 
\end{proof}

In the case $i = 2$, we have the following result, similar to the bound shown in \cite{LMX19}.

We now compare the relative parameters of the codes $C_2(S,D)$ with $D = \left(q\frac{q-1}{2} \right) P_{\infty}$ shown in  Proposition \ref{prop42}.
Recall that	$$R=\frac{k}{n}= \frac{(l+1)}{q^2}  \qquad \text{and} \qquad \delta =\frac{d}{n} \geq \frac{(q-1)(q-2)-l+1}{q^2-q} .$$

\begin{pro}\label{propcota} For each odd $q \geq 5$ and any $1 \leq l\leq  (q-1)(q-2)$, the relative parameters $R$ and $\delta$ of $C_2(S,D)$ with $D = l P_{\infty}$ verify:
	
$$ R + \frac{q-1}{q} \delta > \frac{r}{r+1} \left( \frac{q-2}{q} \right),$$ or equivalently,

\begin{equation}\label{cotanuestra}
	R > \frac{r}{r+1} \left( 1 - \delta - \frac{2}{q} \right),
\end{equation} where $r = q-1$ is the locality of $C_2$.
\end{pro}

\begin{proof}
	Since $\delta = \dfrac{(q-1)(q-2)-l+1}{q^2-q} = \dfrac{q}{q-1} \dfrac{(q-1)(q-2)-l+1}{q^2}$ we have $$R + \frac{q-1}{q} \delta = \frac{(q-1)(q-2)+2}{q^2} > \frac{(q-1)(q-2)}{q^2} = \frac{q-1}{q} \left( 1 - \frac{2}{q} \right).$$ 
\end{proof}

\begin{pro}\label{propejemplo}
	For each odd $q \geq 5$ and $l = q \left(\frac{q-1}{2}-1 \right)$, the relative parameters $R$ and $\delta$ of $C_2(S,D)$ with $D = l P_{\infty}$ verify:
	
	\begin{equation}\label{parametrosnuestros} R= \frac{(l+1)}{q^2} > \frac{1}{2}\left( \frac{q-2}{q} \right)^2  \quad \text{and} \quad \delta = \frac{(q-1)(q-2)-l+1}{q^2-q} > \frac{1}{2} \left( \frac{q-3}{q-1} \right). \end{equation}
\end{pro}

\begin{proof}
	By Proposition \ref{prop42} we have  $n = q^2(q^2-q)$, $k = (l+1)(q-1)q$ and $d = \frac{1}{2}q^2(q^2-3q+6)$ so
	\begin{align*} R &= \frac{l+1}{q^2}  \\
		& = \frac{q \left(\frac{q-1}{2} -1 \right)+1}{q^2} \\
		& = \frac{1}{2} \frac{q (q-1-2)+2}{q^2} \\
		& = \frac{1}{2} \frac{(q-1)(q-2)}{q^2}  > \frac{1}{2} \left( \frac{q-2}{q} \right)^2.
	\end{align*}
On the other hand, 	$$\delta = \frac{1}{2}\frac{q^2-3q+6}{q^2-q} > \frac{1}{2}\frac{q-3}{q-1}.$$

\end{proof}

For $q \geq 5$ and $1 \leq l \leq (q-1)(q-2)$, Proposition \ref{propcota} ensures that the relative parameters of $C_2(S,D)$ are all above the line $
	 R = \frac{q}{q+1}\left( 1 - \delta - \frac{2}{q} \right)$ improving slightly a result in \cite{BargTamoVladut}, while Proposition \ref{propejemplo} shows the existence of a code whose relative parameters lie exactly on that line. In Figures \ref{ejemploq7} and \ref{ejemploq17} below, we can see those cases for $q=7$ and $q=17$.

\begin{figure}[h]
\includegraphics[scale=0.5]{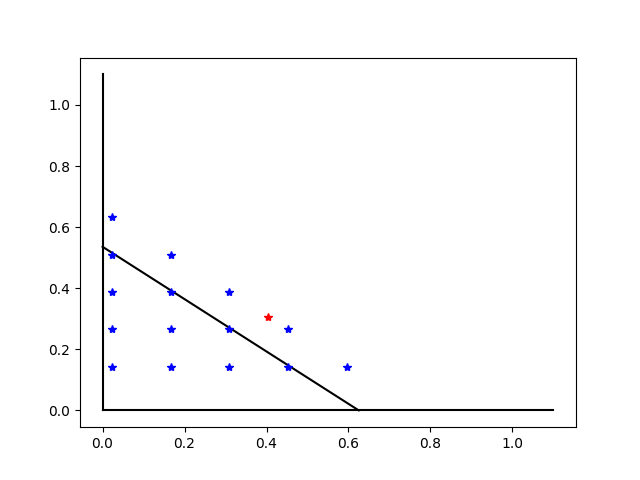}
 \caption{Inequality \eqref{GVthem} for $r+1=q=7$ is shown as a black line. For $2 \leq i \leq 6$, lower bounds of the relative parameters of $C_i(S,D)$ are shown as blue points. The parameters of $C_2(S,D)$ shown in \eqref{parametrosnuestros} appeared as a red point.} \label{ejemploq7}
\end{figure}

\begin{rem}
 For $C_3(S,D)$ and $1 \leq l \leq (q-1)(q-3)$ we can not guarantee by Proposition \ref{propcotai} that its relative parameters $(R,\delta)$ lie above that line $R = \frac{q}{q+1} \left(1 - \delta - \frac{3}{q+1} \right)$. However, taking Corollary \ref{cerosencomun} in consideration, one could expect that the lower bound of $\delta$ may be improved from some step and on, obtaining in this way a stronger lower bound for $\delta$. 
 
\end{rem}

\begin{figure}[ht]
\includegraphics[scale=0.5]{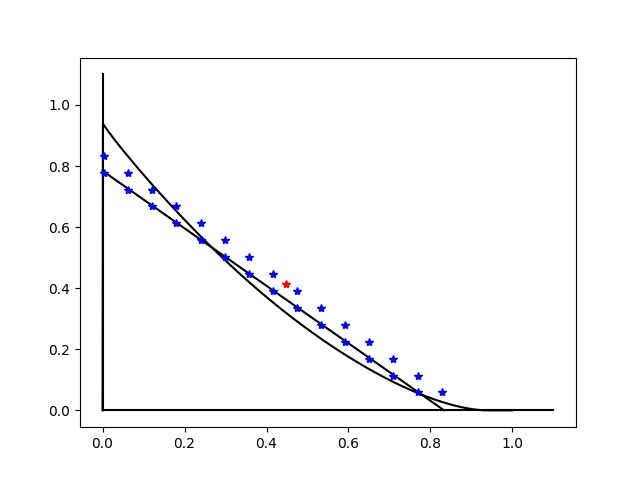}
 \caption{Inequalities \eqref{GVthem} and \eqref{GVbound} for $r+1=q=17$ are shown as black curves. For $2 \leq i \leq 3$, lower bounds of the relative parameters of $C_i(S,D)$ are shown as blue points. The parameters of $C_2(S,D)$ shown in \eqref{parametrosnuestros} appeared as a red point.} \label{ejemploq17}
\end{figure}

\begin{rem}
	For $q$ sufficiently large, and for certain $2 \leq i \leq q-1$, and  $1 \leq l \leq (q-1)(q-i)$, we have from Proposition \ref{propcotai}, that the lower bound for the relative parameters $(R,\delta)$ of $C_i(S,D)$ with $D = l P_{\infty}$, improves a bound analogous to the GV bound, derived in \cite{TBF16}. This is, for suitable $i$ and $l$, we have that $(R,\delta)$ lies above the curve 
	\begin{equation}\label{GVbound} R = \frac{r}{r+1} - \min_{0 < s \leq 1} \left\{	\frac{1}{r+1} \log_qb_2(s) -\delta \log_q(s)\right\}
\end{equation}
where $b_2(s) = \frac{1}{q} \left((1 + (q-1)s)^{r+1}+ (q-1)(1-s)^{r+1} \right).$
\end{rem}

\newpage
\bibliographystyle{plain}

\end{document}